\documentclass[a4paper,10pt]{article}
\usepackage{graphicx} 
\usepackage[utf8]{inputenc}
\usepackage[T1]{fontenc}
\usepackage{comment}
\usepackage{float}
\usepackage{natbib}
\usepackage{pdflscape}  
\usepackage{xcolor}
\usepackage{booktabs, siunitx}

\usepackage{a4wide}
\usepackage{indentfirst}
\usepackage{authblk}

\usepackage{amsmath,amssymb}
\usepackage{amsthm}
\usepackage[colorlinks=true, linkcolor=black, citecolor=black, urlcolor=blue]{hyperref}
\usepackage{footmisc} 
\usepackage{subfig}
\usepackage{graphicx}
\usepackage{subcaption}
\usepackage{chngcntr}
\counterwithin{table}{section}
\counterwithin{figure}{section}

\usepackage[margin=1.2in]{geometry}
\usepackage{graphicx}
\usepackage{listings}
\usepackage[utf8]{inputenc}
\usepackage{appendix}
\usepackage{shuffle}
\numberwithin{equation}{section}
\newtheorem{teo}{Theorem}[section]
\newtheorem*{teo*}{Theorem}
\newtheorem*{prop*}{Proposition}
\newtheorem*{corol*}{Corollary}
\newtheorem{prop}[teo]{Proposition}
\newtheorem{example}[teo]{Example}
\newtheorem{corol}[teo]{Corollary}

\theoremstyle{definition}

\newtheorem{remark}[teo]{Remark}

\newcommand{\Var}{\operatorname{Var}}

\newcommand{\blfootnote}[2]{%
  \begingroup
  \renewcommand\thefootnote{#1}%
  \addtocounter{footnote}{1}%
  \footnotetext{#2}%
  \addtocounter{footnote}{-1}%
  \endgroup
}

\title{Analytic approximation for Bachelier option prices and applications}

\author{Elisa Alòs\textsuperscript{*}, Òscar Burés\textsuperscript{†}}
\date{\today}

\begin{document}

\maketitle

\blfootnote{*}{Department of Economics and Business, Universitat Pompeu Fabra and Barcelona School of Economics. Ramón Trias Fargas 25-27, 08005, Barcelona, Spain.}
\blfootnote{†}{Departament de Matemàtica Econòmica, Financera i Actuarial, Universitat de Barcelona. Diagonal 690--696, 08034 Barcelona, Spain.}

\begingroup
\renewcommand{\thefootnote}{}  
\renewcommand{\footnotemargin}{0pt}  
\footnotetext{%
\noindent\hspace{0pt}Òscar Burés supported by program AGAUR-FI ajuts (2025 FI-1 00580) from the Department of Research and Universities of the Government of Catalonia and the co-funding of the European Social Fund Plus (ESF+). }
\endgroup

\begin{abstract}
    It is well-known that, in the Bachelier model, when asset prices and volatilities are uncorrelated, the implied volatility coincides with the fair value of the volatility swap. In this paper, via classical It\^o calculus and Taylor expansions, we write the price for out-of the-money (OTM) and in-the-money (ITM) options as an expansion with respect to the moneyness, where the coefficients are related to the negative (non-integer) powers of the future mean volatility. As an a application, we use it as a control variate to reduce the variance of Monte Carlo option prices in the correlated case.
\end{abstract}

\section{Introduction}
Today, option pricing theory is based largely on the Black-Scholes model, in which asset prices are log-normal. Most of the popular models in the financial industry (such as local or stochastic volatility models) are extensions of it. In this framework, asset prices are positive. This hypothesis is not always satisfied (as has recently been registered for interest rates or commodities). Then, in some scenarios, markets have moved to the Bachelier model (see \cite{ASENS_1900_3_17__21_0} and \cite{choi2022black}), where asset prices are assumed to be normal.

One of the main problems in option pricing (both in the Black-Scholes or in the Bachelier framework) is the construction of adequate closed-form approximation formulas for option prices and implied volatilities. Towards this end, several works are devoted to constructing expansions in which the leading term is the Black-Scholes/Bachelier price evaluated at a proxy for the implied volatility, that is usually the spot volatility or the variance swap. One classical approach  relies on the analysis of the corresponding PDE with respect to a specific model parameter (see, among others,  \cite{lewis2000option}, \cite{hagan2002managing}, \cite{fouque2000derivatives}, and \cite{fouque2}. Other researchers follow a probabilistic approach, where option prices depend on the joint distribution of the variance swap and asset prices (see, for example,  \cite{antonelli2006pricing}, \cite{fukasawa2011asymptotic}, \cite{BergomiGuyon2012}, \cite{alos2012decomposition}, and \cite{alos2020exponentiation}). The results obtained in these latest works are very general and can be applied when the volatility is not Markovian, as in the case of rough volatilities. Some specific works on the Bachelier implied volatility include \cite{baviera2025smile}, \cite{unkwnown}, \cite{alos2025short}, and the references therein. 

In all the above papers, the expansion contains a first correction term due to the correlation (associated with the leverage swap), a second one due to the vol-of-vol (associated with the quadratic variation of the variance swap), and higher-order terms. Even when these approximations work well near at-the-money strikes, they are not analytical (see, for example, \cite{Lewis02092022}), and their region of validity is limited.

Our purpose in this paper is to obtain an analytical expansion for Bachelier option prices, in the case of uncorrelated asset prices and volatilities. Via adequate decomposition formulas, we write the option price as the ATM price plus a correction due to the moneyness. Then, a Taylor expansion allows us to write this correction in terms of powers of the moneyness, with coefficients depending on negative (non-integer) powers of the future integrated volatility. 

Our numerical examples on the SABR and the Heston model confirm the validity of this approximation. As an application, we use it as a control variate in the simulation of option prices. This technique leads to a significant variance reduction in the Monte Carlo option pricing.

\section{Preliminaries}
We consider the Bachelier model for asset prices under a risk-neutral probability $P$:
\begin{equation}
\label{Bachelier}
dX_t=\sigma _{t}\left( \rho dW_{t}+\sqrt{1-\rho ^{2}}%
B_{t}\right) ,\,t\in \lbrack 0,T] 
\end{equation}
for some $T>0$, where $W$ and $B$ are independent standard Brownian motions, $\rho \in \left[-1,1\right],$ and $\sigma$ is a square integrable
process adapted to the filtration generated by the Brownian motion $%
W$. As in the previous chapters, we denote by $\mathcal{F}^W$ and $\mathcal{F}^B$ the filtrations generated by $W$ and $ B$, respectively, and 
$\mathcal{F}:=\mathcal{F}^W\vee\mathcal{F}^B$. If $\sigma$ is constant and $\rho=0$, the above model is called the {\bf Bachelier model}.

We denote by $Bac(T,x,k,\sigma)$  the classical Bachelier price of a European call with time to maturity \(T\), current stock price \(x\),
strike price \(k\) and volatility \(\sigma\). That is, $$
  Bac(T,x,k,\sigma)=(x-k)N(d_{Bac}(k,\sigma))+N'(d_{Bac}(k,\sigma))\sigma\sqrt{T},$$
with
$$
  d_{Bac}(k,\sigma)=\frac{x-k}{\sigma\sqrt{T}},$$ where \(N\) is the cumulative distribution
function and the probability density  function of the standard normal random variable. 

\vspace{0.5cm}

We denote by $\mathcal{L}_{Bac}\left( \sigma \right) $ denotes the Bachelier 
differential operator with volatility $\sigma :$%
\begin{equation*}
\mathcal{L}_{Bac}\left( \sigma \right) =\frac{\partial }{\partial t}+\frac{1}{%
2}\sigma ^{2}\frac{\partial ^{2}}{\partial x^{2}}
\end{equation*}
It is well known that $\mathcal{L}_{Bac}\left( \sigma \right) Bac\left( \cdot
,\cdot,\cdot ;\sigma \right) =0.$

\vspace{0.5cm}

Finally, we define the Bachelier implied volatility of a traded call option $I^{Bac}(k)$  as the unique volatility parameter one should put in the Bachelier formula to get the market option price $V$. That is, the quantity $I^{Bac}(k)$ such that
$$
V=Bac(T,X_{0},k,I^{Bac}(k)),
$$
where $X_0$ denotes the asset price and $k$ the strike price of the option. Notice that, if $k=X_0$,
\begin{equation}
\label{Ibac}
V=Bac(T,X_{0},X_0,I^{Bac}(X_0))=N'(0)I^{Bac}\sqrt{T}=\frac{1}{\sqrt{2\pi}}I^{Bac}(X_0)\sqrt{T}.
\end{equation}
At the same time, due to the definition of the Black-Scholes implied volatility,
\begin{equation}
\label{Ibs}
V=BS(T,X_0,X_0,I(X_0))=X_0\left(2N\left(  \frac{I(X_0)\sqrt{T}}{2} \right)-1\right).
\end{equation}
Then, (\ref{Ibac}) and (\ref{Ibs}) imply the following conversion formula for ATM implied volatilities:
\begin{equation}
\label{Ibacbsatm}
I^{Bac}(X_0)=\frac{\sqrt{2\pi}}{\sqrt{T}}X_0\left(2N\left(  \frac{I(X_0)\sqrt{T}}{2} \right)-1\right)
\end{equation}
We will also need the following notations. 
\begin{itemize}
\item $v=\sqrt{\frac{1}{T}E\int_0^T\sigma_s^2 ds} $ is the square root of the variance swap.
\item $\hat{v}=E\sqrt{\frac{1}{T}\int_0^T\sigma_s^2 ds} $ is the volatility swap.
\item For all $s\in [0,T]$, we define $M_s=\frac{1}{T}E_s\int_0^T\sigma_u^2ds$.
\item For all $s\in [0,T]$, we denote $v_{s}=\sqrt{\frac{1}{T}E_s\int_0^T\sigma_u^2 du}$. In particular, $v_{0}=v.$
\end{itemize}
Notice that $v=\sqrt{M_0}$, $v_s=\sqrt{M_s}$, and $\hat{v}=E\sqrt{M_T}$. Then, a direct application of It\^o's formula to the process $M$ and the function $f(x)=\sqrt{x}$ leads to the following relationship between the variance and the volatility swap
\begin{equation}
\label{convexity}
\hat{v}=v-\frac{1}{8}E\int_0^T\frac{1}{v_{s}^3}d\langle M,M\rangle_s
\end{equation}
\section{An analytical expansion for option prices}
Our approach is based on the following decomposition for option prices in the uncorrelated case. We assume the following integrability condition.

(H) For all $p>1$, $v^{-1}$ and $|\frac{d\langle M,M\rangle_s}{ds}|$ are in $L^p ([0,T]\times \Omega)$. 

\begin{prop}[Decomposition formula for option prices in the uncorrelated case]\label{decomp_formula0}
Consider the model (\ref{Bachelier}) with $\rho=0$ and assume that Hypothesis (H) holds. Then
\begin{eqnarray*}
V&=&Bac(T,X_{0},k,v) \\
&&+\frac{T^2}{8} \,E\left( \int_{0}^{T}K_{Bac}(T,X_0,k,v_s)d\langle
M,M\rangle_s \right),
\end{eqnarray*}
where 
\begin{eqnarray*}
K_{Bac}(T,x,\sigma) & = &
\frac{\partial^4 Bac}{\partial x^4}(T,x,\sigma) \\
& = &  \frac{(x-k)^2-T\sigma^2}{T^{\frac52}\sigma^5}\frac{\exp \left(-\frac{d_{Bac}^2(\sigma)}{2}\right)}{ \sqrt{2\pi}}%
\end{eqnarray*}
\end{prop}
\begin{proof}
Using conditional expectations, classical arguments allow us to write the option price $V$ as
$$V=E(Bac(T,X_0,k,v_{T})$$
Now, a direct application of It\^o's formula and the fact that
$$
\frac{\partial Bac}{\partial \sigma}(T,X_{0},k,\sigma)\frac{1}{\sigma T}=\frac{\partial^2 Bac}{\partial x^2}(T,X_{0},k,\sigma)
$$
give us that
\begin{eqnarray}
Bac(T,X_{0},k,v_{T})&=&Bac(T,X_{0},k,v_{T})\nonumber\\
&+&\frac12 T\int_0^T\frac{\partial^2 Bac}{\partial x^2}(T,X_{0},k,v_s)dM_s\nonumber\\
&+&\frac18 T^2\int_0^T\frac{\partial^4 Bac}{\partial x^4}(T,X_{0},k,v_s)d\langle M,M\rangle_s.\nonumber
\end{eqnarray}
Then, taking expectations, and taking into account that $v_{T}=v$, we get
\begin{eqnarray}
V&=&Bac(T,X_{0},k,v)\nonumber\\
&+&\frac18 T^2E\int_0^T\frac{\partial^4 Bac}{\partial x^4}(T,X_{0},k,v_s)d\langle M,M\rangle_s,\nonumber
\end{eqnarray}
and now the proof is complete.
\end{proof}
As a direct corollary, we get the following decomposition formula
\begin{corol}
    Assume the model (\ref{Bachelier}) and assume that hypothesis (H) holds. Then

 \begin{eqnarray*}
        \label{segona}
V =Bac(T,X_{0},k,v)&+&\frac{(X_0-k)^2}{8T^{\frac12}\sqrt{2\pi}}E\int_0^T \frac{1}{v_s^5}\exp \left(-\frac{d_{Bac}^2(v_s)}{2}\right)d\langle M,M\rangle_s\nonumber\\
&-&\frac{1}{8} \frac{\sqrt{T}}{\sqrt{2\pi}}E\left( \int_{0}^{T}\exp \left(-\frac{d_{Bac}^2(\sigma)}{2}\right)\frac{1}{v_s^3}d\langle
M,M\rangle_s \right).
\end{eqnarray*}
\end{corol}
\begin{proof}
Notice that 
\begin{eqnarray*}
K_{Bac}(T,x,k,\sigma) & = &   \frac{(x-k)^2-T\sigma^2}{T^{\frac52}\sigma^5}\frac{\exp \left(-\frac{d_{Bac}^2(\sigma)}{2}\right)}{ \sqrt{2\pi}}\nonumber\\
&=&\frac{(x-k)^2}{T^{\frac52}\sigma^5}\frac{\exp \left(-\frac{d_{Bac}^2(\sigma)}{2}\right)}{ \sqrt{2\pi}}\nonumber\\
&-&\frac{1}{T^\frac32\sigma^3\sqrt{2\pi}}\exp \left(-\frac{d_{Bac}^2(\sigma)}{2}\right).
\end{eqnarray*}
Now, as 
$$\frac{\partial Bac}{\partial \sigma}(T,x,k,v_s)
=\frac{\sqrt{T}}{\sqrt{2\pi}}\exp \left(-\frac{d_{Bac}^2(\sigma)}{2}\right)$$
it follows that 
\begin{eqnarray*}
K_{Bac}(T,x,k,\sigma) & = &   \frac{(x-k)^2}{T^{\frac52}\sigma^5}\frac{\exp \left(-\frac{d_{Bac}^2(\sigma)}{2}\right)}{ \sqrt{2\pi}}\nonumber\\
&-&\frac{1}{T^2 \sigma^3}\frac{\partial Bac}{\partial \sigma}(T,x,k,v_s).
\end{eqnarray*}
Then, Proposition \ref{decomp_formula0}
leads to

\begin{eqnarray*}
V =Bac(T,X_{0},k,v)&+&\frac{(X_0-k)^2}{8T^{\frac12}\sqrt{2\pi}}E\int_0^T \frac{1}{v_s^5}\exp \left(-\frac{d_{Bac}^2(v_s)}{2}\right)d\langle M,M\rangle_s\nonumber\\
&-&\frac{1}{8} \frac{\sqrt{T}}{\sqrt{2\pi}}E\left( \int_{0}^{T}\exp \left(-\frac{d_{Bac}^2(\sigma)}{2}\right)\frac{1}{v_s^3}d\langle
M,M\rangle_s \right).
\end{eqnarray*}

\end{proof}
\begin{remark}[The ATMI and the volatility swap]
 Notice that, if $k=x$, $\frac{\partial Bac}{\partial \sigma}(T,x,k,v_s)$ is deterministic and then
$$Bac(T,X_0,X_0,\hat{v})=Bac(T,X_{0},X_0,v)-\frac{1}{8} \,E\left( \int_{0}^{T}\frac{\partial Bac}{\partial \sigma}(T,X_0,X_0,v_s)\frac{1}{v_s^3}d\langle
M,M\rangle_s \right),$$
which implies that, for ATM options, $V=Bac(T,X_0,X_0,\hat{v})$, according to the well-known properties of the Bachelier implied volatility.
\end{remark}
Now we are in a position to prove the main result of this paper.
\begin{teo}[Price expansion]
\label{principal}
Consider the model (\ref{Bachelier}) with $\rho=0$ and assume that Hypothesis (H) holds. Then    

\begin{eqnarray}
    V&=&Bac(T,X_0,k,v)+\frac{T^\frac12}{\sqrt{2\pi}}(\hat{v}-v)\nonumber\\
&-&\frac{T^{\frac12}}{2\sqrt{2\pi}}\sum_{n=1} \frac{1}{n!(2n-1)} \left(-\frac{(X_0-k)^2}{2 T}\right)^n\nonumber\\
&&\times\left[
E\left(\frac{1}{T}\int_0^T \sigma_s^2ds\right)^{\frac12-n}-\left(\frac{1}{T}\int_0^T E (\sigma_s^2)ds\right)^{\frac12-n}\right]\nonumber\\
\end{eqnarray}
\end{teo}
\begin{proof}
    Equation (\ref{segona}) gives us that
\begin{eqnarray}
V=Bac(T,X_0,k,v)&-&\frac{T^{\frac12}}{4\sqrt{2\pi}}E\int_0^T \frac{1}{v_s^3} \left(-\frac{(X_0-k)^2}{2v_s^2 T}\right)\exp \left(-\frac{(X_0-k)^2}{2v_s^2 T}\right)d\langle M,M\rangle_s\nonumber\\
&-&\frac{1}{8} \frac{\sqrt{T}}{\sqrt{2\pi}}E\left( \int_{0}^{T}\exp \left(-\frac{(X_0-k)^2}{2v_s^2 T}\right)\frac{1}{v_s^3}d\langle
M,M\rangle_s \right)\nonumber\\
=Bac(T,X_0,k,v)&-&\frac{T^{\frac12}}{4\sqrt{2\pi}}
\sum_{n=1} \frac{1}{(n-1)!} \left[\left(-\frac{(X_0-k)^2}{2 T}\right)^n E\int_0^T \frac{1}{v_s^{3+2n}} d\langle M,M\rangle_s\right]\nonumber\\
&-&\frac{1}{8} \frac{T^\frac12}{\sqrt{2\pi}}\sum_{n=0} \frac{1}{n!}\left[\left(-\frac{(X_0-k)^2}{2v_s^2 T}\right)^n E\int_{0}^{T} \frac{1}{v_s^{3+2n}}d\langle
M,M\rangle_s \right]\nonumber\\
=Bac(T,X_0,k,v)&-&\frac{1}{8} \frac{T^\frac12}{\sqrt{2\pi}} E\int_{0}^{T} \frac{1}{v_s^{3}}d\langle
M,M\rangle_s \nonumber\\
&-&\frac{T^{\frac12}}{\sqrt{2\pi}}
\sum_{n=1} \left(\frac{1}{4(n-1)!} +\frac{1}{8n!}\right)\left[\left(-\frac{(X_0-k)^2}{2 T}\right)^n E\int_0^T \frac{1}{v_s^{3+2n}} d\langle M,M\rangle_s\right]\nonumber \\
Bac(T,X_0,k,v)&-&\frac{1}{8} \frac{T^\frac12}{\sqrt{2\pi}} E\int_{0}^{T} \frac{1}{v_s^{3}}d\langle
M,M\rangle_s \nonumber\\
&-&\frac{T^{\frac12}}{\sqrt{2\pi}}
\sum_{n=1}\frac{2n+1}{8n!}\left[\left(-\frac{(X_0-k)^2}{2 T}\right)^n E\int_0^T \frac{1}{v_s^{3+2n}} d\langle M,M\rangle_s\right]\nonumber
\end{eqnarray}
Now, notice that, for all real $\theta$ 
$$
E(M_T^{\theta/2})=M_0^{\theta/2 }+\frac12\frac{\theta}{2}\left(\frac{\theta}{2}-1\right)E\int_0^T v_s^{\left(\theta-4\right)} d\langle M,M\rangle_s.
$$
Now, taking $-3-2n=\theta -4$ we have
$\theta=1-2n$ and then $\left(\frac{\theta}{2}-1\right)=n^2-0.25$. 
This implies that
\begin{eqnarray}
E\int_0^T \frac{1}{v_s^{3+2n}} d\langle M,M\rangle_s&=&\frac{2}{n^2-0.25}\left(E(M_T^{1/2-n})-M_0^{1/2-n}\right)\\
&=&\frac{2}{n^2-0.25}\left(
E\left(\frac{1}{T}\int_0^T \sigma_s^2ds\right)^{\frac12-n}-\left(\frac{1}{T}\int_0^T E \sigma_s^2ds\right)^{\frac12-n}\right),\nonumber
\end{eqnarray}
and now the proof is complete.
\end{proof}
\begin{remark}
The above result is an exact equality. It is an analytic expansion, and then it is valid for strikes and maturities. It reduces the computation of option prices $Bac(T,X_0,k, v_T)$ to the estimation of the corresponding negative (non-integer) moments of $M_T$. Once these moments are obtained and stored, the calculation of the option price in any concrete strike $k$ is obtained via a closed-form formula.
\end{remark}
\begin{remark} \label{Greeks}
Theorem \ref{principal} does not only give an expression for the option price, but it also allows us to deduce, taking derivatives, an analytical formula for the Greeks. For example, the Delta $\Delta$ of a call is given by

\begin{eqnarray}
    \Delta&=&N(d_{Bac}(v))\nonumber\\
&+&\frac{(X_0-k)}{\sqrt{2T\pi}}\sum_{n=1}\frac{1}{(n-1)!(2n-1)} \left(-\frac{(X_0-k)^2}{2 T}\right)^{n-1}\nonumber\\
&&\times\left[
E\left(\frac{1}{T}\int_0^T \sigma_s^2ds\right)^{\frac12-n}-\left(\frac{1}{T}\int_0^T E (\sigma_s^2)ds\right)^{\frac12-n}\right].\nonumber\\
\end{eqnarray}
and the following expression for the Gamma $\Gamma$ holds:

\begin{eqnarray}
   \Gamma &=&\frac{1}{\sqrt{2\pi T}v} e^{-d_{Bac}^2(v)/2}\nonumber\\
&+&\frac{1}{\sqrt{2T\pi}}\sum_{n=1} \frac{1}{(n-1)!} \left(-\frac{(X_0-k)^2}{2 T}\right)^{n-1}\nonumber\\
&&\times\left[
E\left(\frac{1}{T}\int_0^T \sigma_s^2ds\right)^{\frac12-n}-\left(\frac{1}{T}\int_0^T E (\sigma_s^2)ds\right)^{\frac12-n}\right].
\end{eqnarray}

\end{remark}
\begin{remark}
An analytical expression for the Bachelier price in the uncorrelated case can be the starting point for several applications. In the next section, we will see how to use it as a control variate in the Monte Carlo computation of option prices in the correlated case.
\end{remark}
\section{Numerical examples}
\begin{example}[The Heston model ]
Let us assume a Heston-Bachelier model where the volatility process is given by
\begin{equation}
d\sigma_t^2=-\kappa (\sigma_t^2-\theta)+\nu\sqrt{\sigma_t^2} dB_t,
\end{equation}
where $\kappa, \theta,$ and $\nu$ are positive real numbers.
Then, a straightforward computation leads to 
$$M_0=\theta  + \frac{\sigma^2-\theta}{\kappa T}  \left( 1 - e^{-\kappa T} \right).$$
Consider the parameters $\sigma_0 = 20, \kappa = 2, \theta = 400,$ and $ \nu = 20$. The first thing we will explore is how good does the approximation given by Theorem \ref{principal} work versus a benchmark. As a benchmark, we have chosen uncorrelated Call prices with initial asset price $X_0 = 100$, maturities $T \in \{0.8, 1.0, 1.2\}$ and strikes $k \in [70, 140]$. The values of such options are computed with 100,000 conditional Monte Carlo simulations with antithetic variables. For the expansion, we have chosen $N = 30$ as the number of terms of the expansion. In figures \ref{fig:heston_prices} we see how our approximation fits accurately the option prices. In order to confirm the high accuracy of our approximation, in Figure \ref{fig: heston ivs} we see that the great option price fitting is also translated in a highly accurate fit of the implied volatility smiles.
\begin{figure}[h]
    \centering
    \begin{minipage}{0.48\textwidth}
        \centering
        \includegraphics[width=\textwidth]{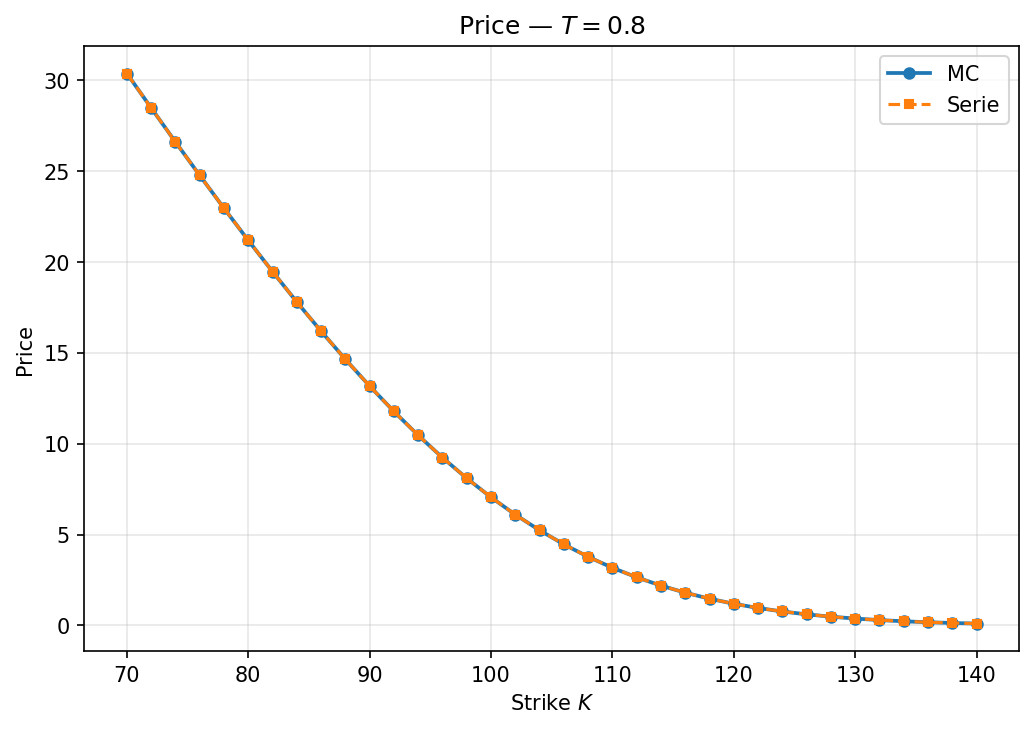}
    \end{minipage}
    \hfill
    \begin{minipage}{0.48\textwidth}
        \centering
        \includegraphics[width=\textwidth]{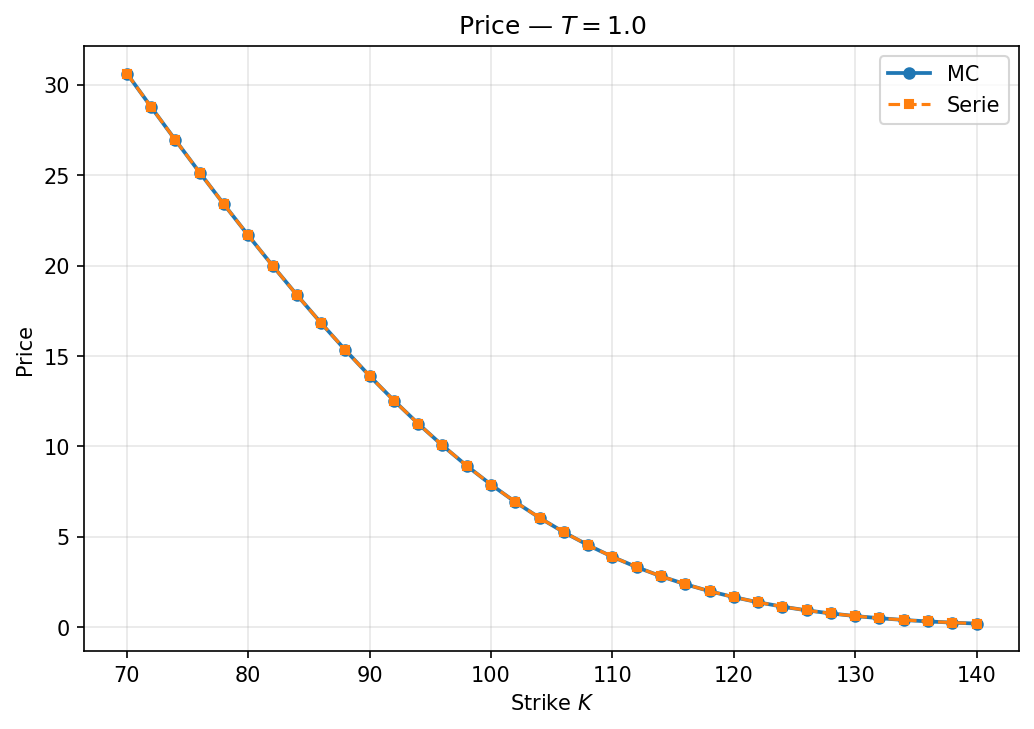}
    \end{minipage}

    \vspace{0.4cm}

    \begin{minipage}{0.48\textwidth}
        \centering
        \includegraphics[width=\textwidth]{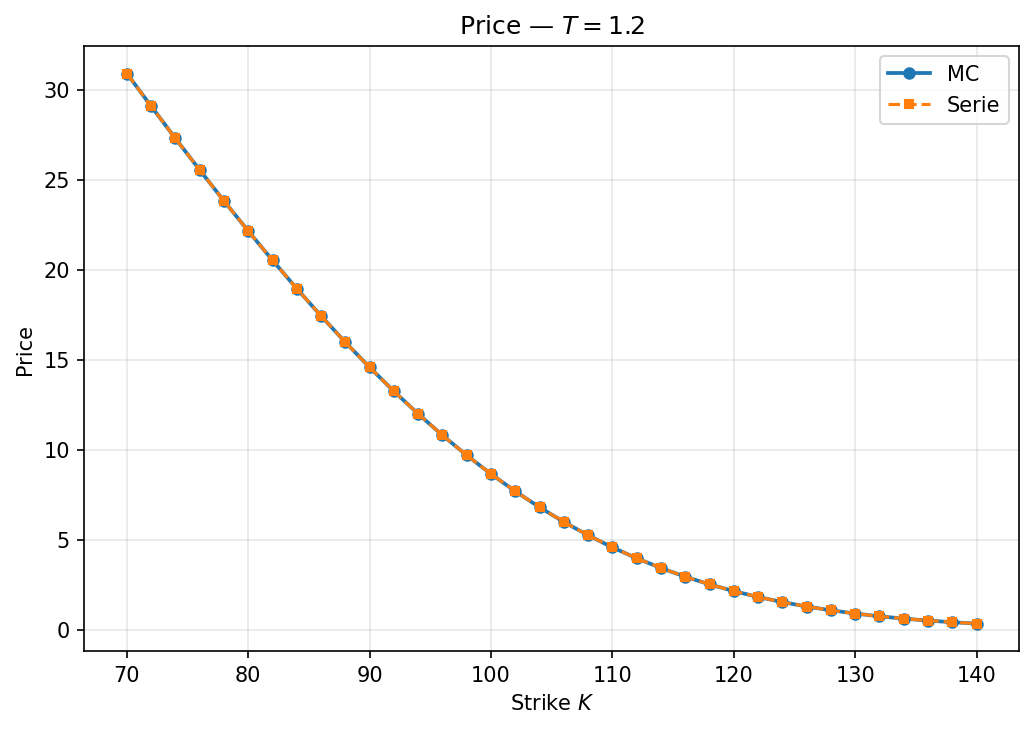}
    \end{minipage}

    \caption{Approximation of option prices for the Heston model.}
    \label{fig:heston_prices}
\end{figure}
\begin{figure}[H]
    \centering
    \begin{minipage}{0.48\textwidth}
        \centering
        \includegraphics[width=\textwidth]{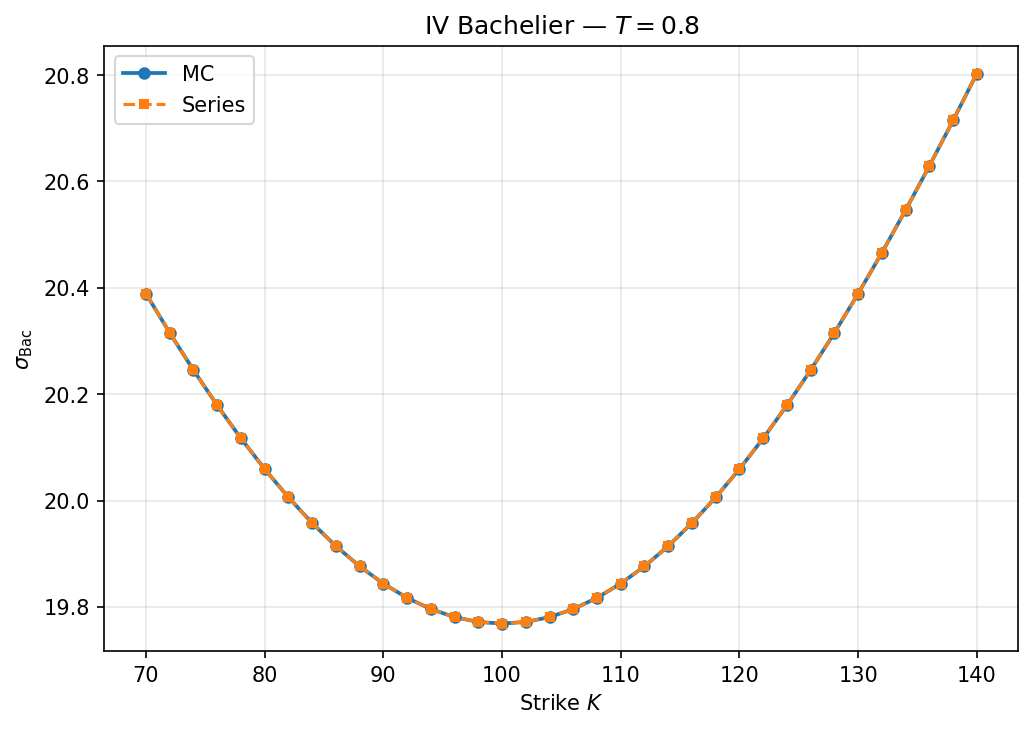}
    \end{minipage}
    \hfill
    \begin{minipage}{0.48\textwidth}
        \centering
        \includegraphics[width=\textwidth]{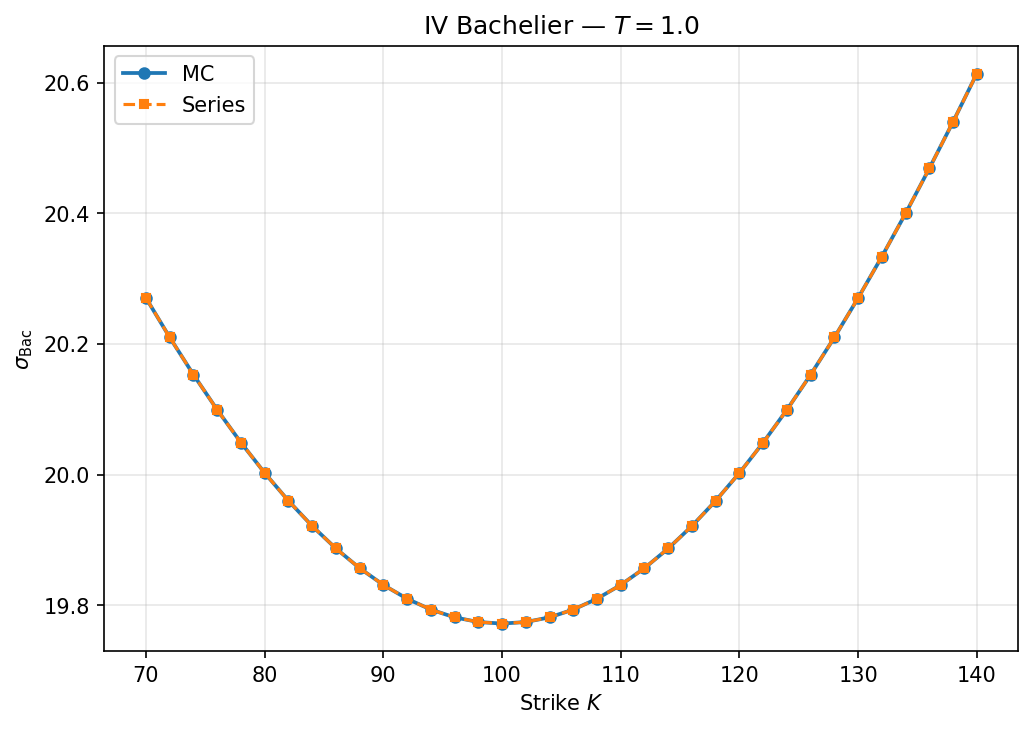}
    \end{minipage}

    \vspace{0.4cm}

    \begin{minipage}{0.48\textwidth}
        \centering
        \includegraphics[width=\textwidth]{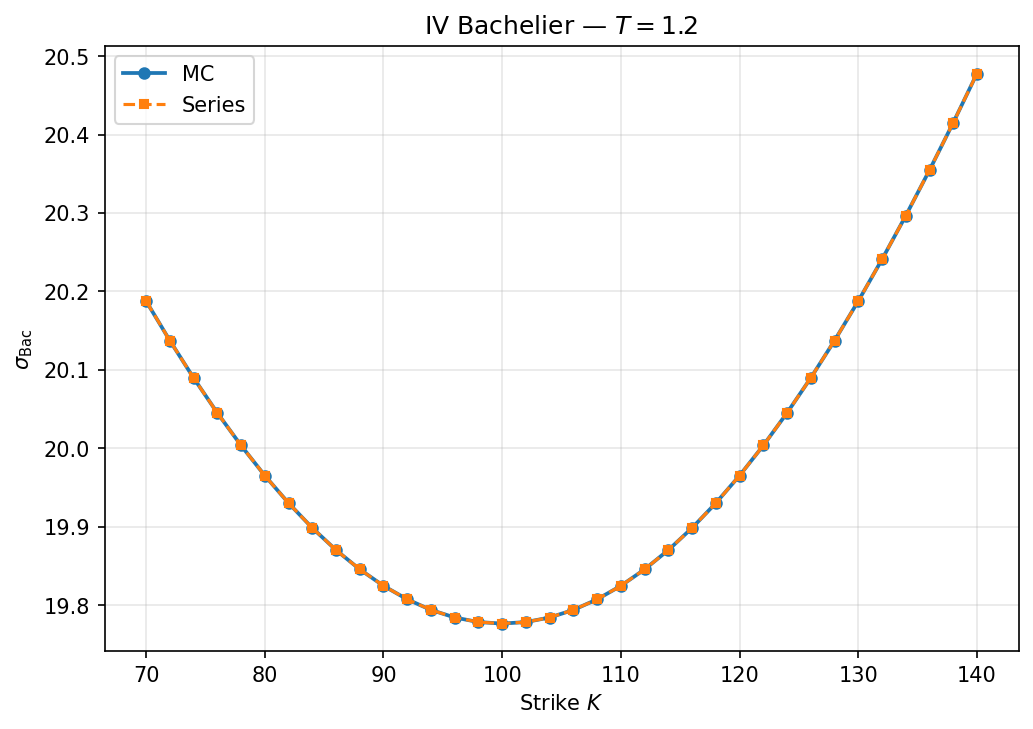}
    \end{minipage}

    \caption{Approximation of implied volatilities for the Heston model.}
    \label{fig: heston ivs}
\end{figure}
Since the precision for $N = 20$ is quite high, one can wonder how many terms are needed to obtain a certain level of accuracy. In order to answer this question, we have found the minimum number $N^*$ such that the error between the implied volatilites computed with $N^*$ and $N^*+1$ terms is less than $0.01$. In a sense, $N^*$ denotes the term in which adding more terms does not substantially change the approximation of the implied volatility. In Tables \ref{tab:noptim_T08}, \ref{tab:noptim_T10} and \ref{tab:noptim_T12} we detail such "optimal" number of terms $N^*$ for a selection of the options used for the implied volatility fitting.
\begin{table}[H]
    \centering
    \begin{tabular}{ccc}
        \hline
        Strike & $N^*$ & Error with next term \\
        \hline
        70  &  9 & 0.004039 \\
        78  &  5 & 0.004954 \\
        86  &  3 & 0.001626 \\
        94  &  1 & 0.002159 \\
        102 &  1 & 0.000025 \\
        110 &  2 & 0.001593 \\
        118 &  4 & 0.002405 \\
        126 &  7 & 0.003052 \\
        134 & 12 & 0.004557 \\
        \hline
    \end{tabular}
    \caption{Optimal number of terms and error for $T = 0.8$}
    \label{tab:noptim_T08}
\end{table}

\begin{table}[H]
    \centering
    \begin{tabular}{ccc}
        \hline
        Strike & $N^*$ & Error with next term \\
        \hline
        70  & 7 & 0.005036 \\
        78  & 4 & 0.006216 \\
        86  & 2 & 0.006591 \\
        94  & 1 & 0.001379 \\
        102 & 1 & 0.000016 \\
        110 & 2 & 0.000755 \\
        118 & 3 & 0.005396 \\
        126 & 5 & 0.009715 \\
        134 & 9 & 0.005078 \\
        \hline
    \end{tabular}
    \caption{Optimal number of terms and error for $T = 1.0$}
    \label{tab:noptim_T10}
\end{table}

\begin{table}[H]
    \centering
    \begin{tabular}{ccc}
        \hline
        Strike & $N^*$ & Error with next term \\
        \hline
        70  & 6 & 0.004298 \\
        78  & 4 & 0.002175 \\
        86  & 2 & 0.003477 \\
        94  & 1 & 0.000919 \\
        102 & 1 & 0.000011 \\
        110 & 1 & 0.007597 \\
        118 & 3 & 0.002329 \\
        126 & 5 & 0.002575 \\
        134 & 7 & 0.008735 \\
        \hline
    \end{tabular}
    \caption{Optimal number of terms and error for $T = 1.2$}
    \label{tab:noptim_T12}
\end{table}
\end{example}

\begin{example}[The SABR model]
Let us consider the SABR model where
$$\sigma_t=\sigma_0\exp\left(-\frac{\nu^2}{2} t+\nu B_t\right).$$
Then a direct computation leads to 
$$
M_0=\sigma_0^2 \left(\frac{\exp{\nu^2 T}-1}{\nu^2 T}\right).
$$

Consider the parameters $\sigma_0 = 20$, and $ \nu = 0.5$. In the following plots, we can see the goodness of approximation of the series for option prices and implied volatilities. As before, the benchmark has been obtained from 100,000 Monte Carlo simulations with antithetic variables, and for the expansion we have taken $N=30$ terms. In Figures \ref{fig:SABRprices} and \ref{fig:SABRivs} we can see how does our method fid the option prices and the implied volatilities.
\begin{figure}[H]
    \centering
    \begin{minipage}{0.48\textwidth}
        \centering
        \includegraphics[width=\textwidth]{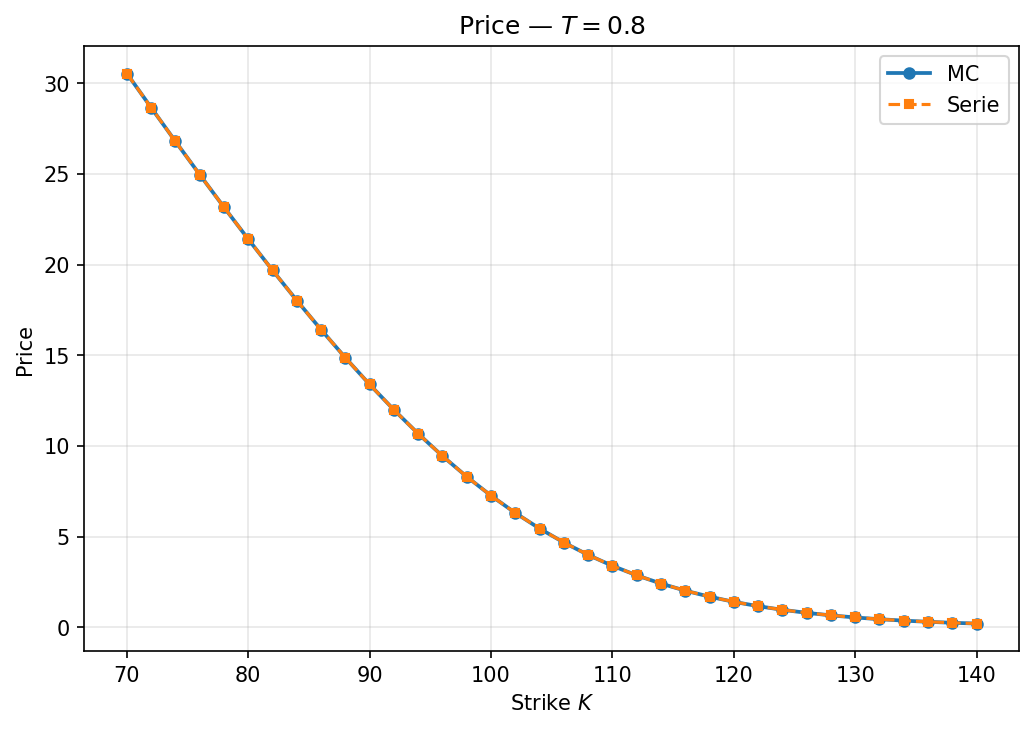}
    \end{minipage}
    \hfill
    \begin{minipage}{0.48\textwidth}
        \centering
        \includegraphics[width=\textwidth]{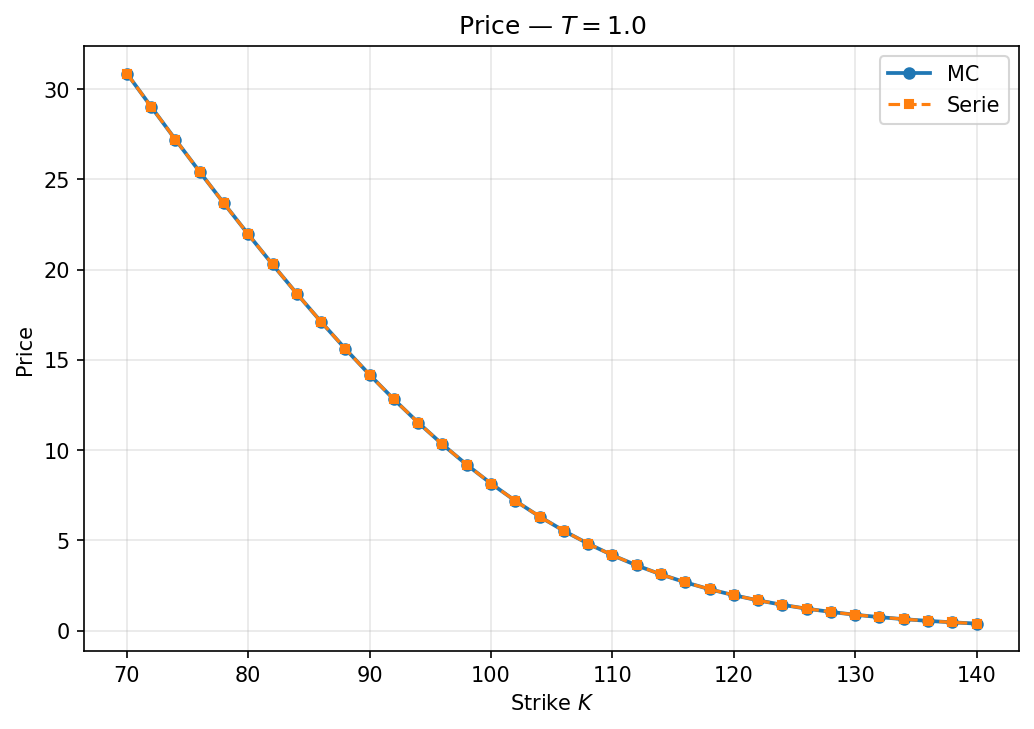}
    \end{minipage}

    \vspace{0.4cm}

    \begin{minipage}{0.48\textwidth}
        \centering
        \includegraphics[width=\textwidth]{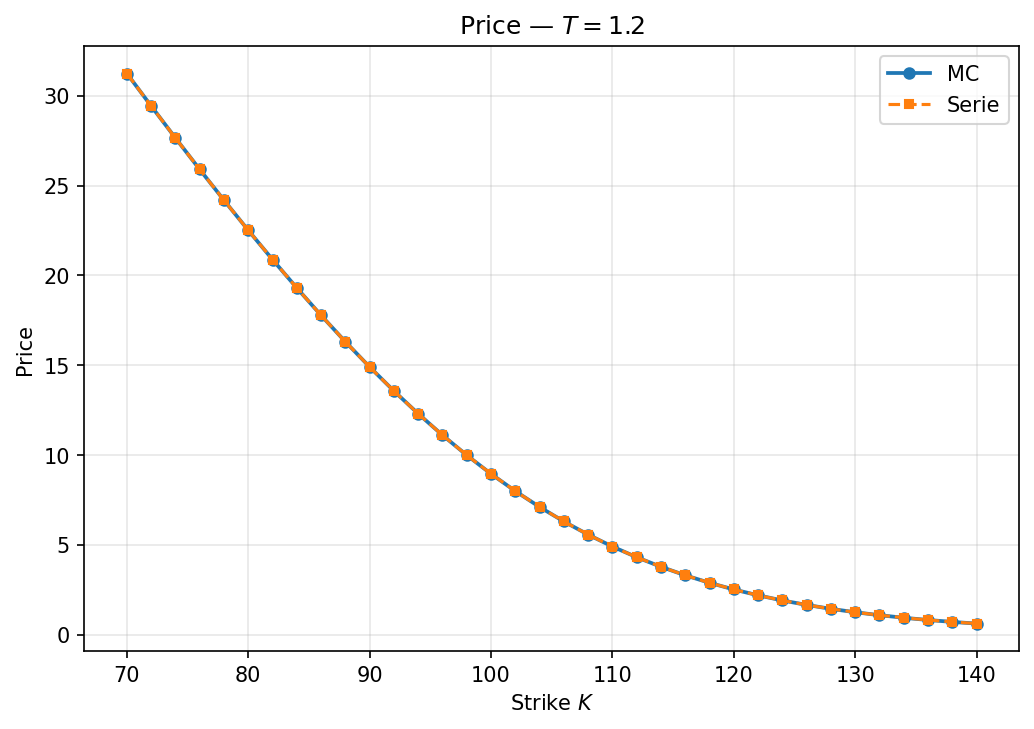}
    \end{minipage}

    \caption{Approximation of option prices for the SABR model}
    \label{fig:SABRprices}
\end{figure}

\begin{figure}[H]
    \centering
    \begin{minipage}{0.48\textwidth}
        \centering
        \includegraphics[width=\textwidth]{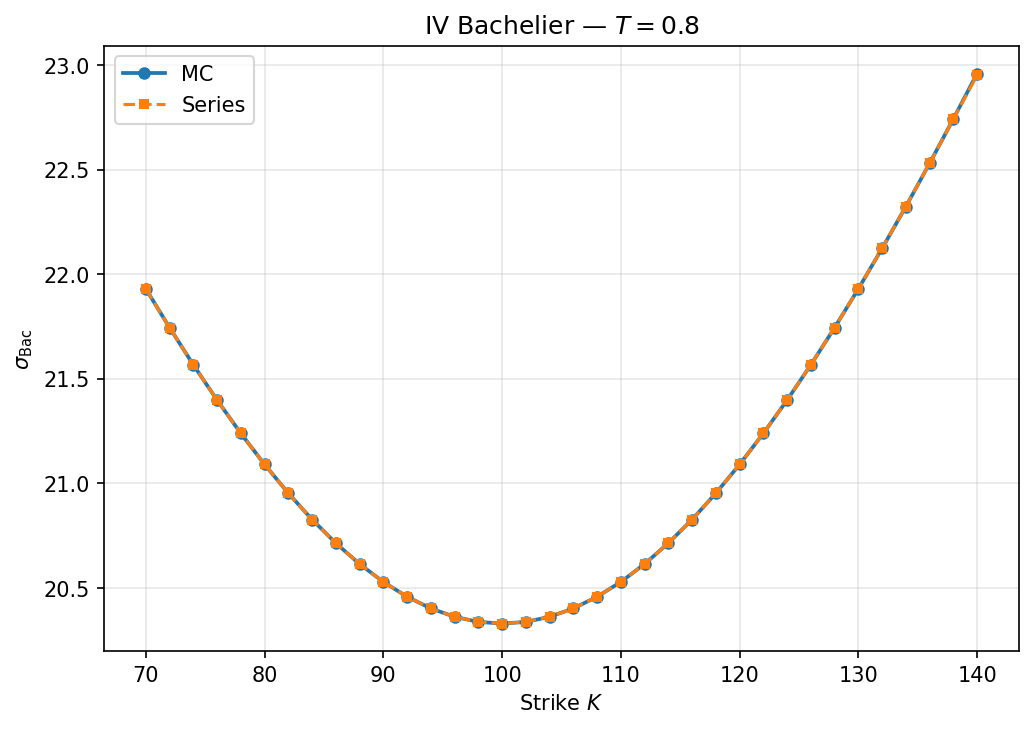}
    \end{minipage}
    \hfill
    \begin{minipage}{0.48\textwidth}
        \centering
        \includegraphics[width=\textwidth]{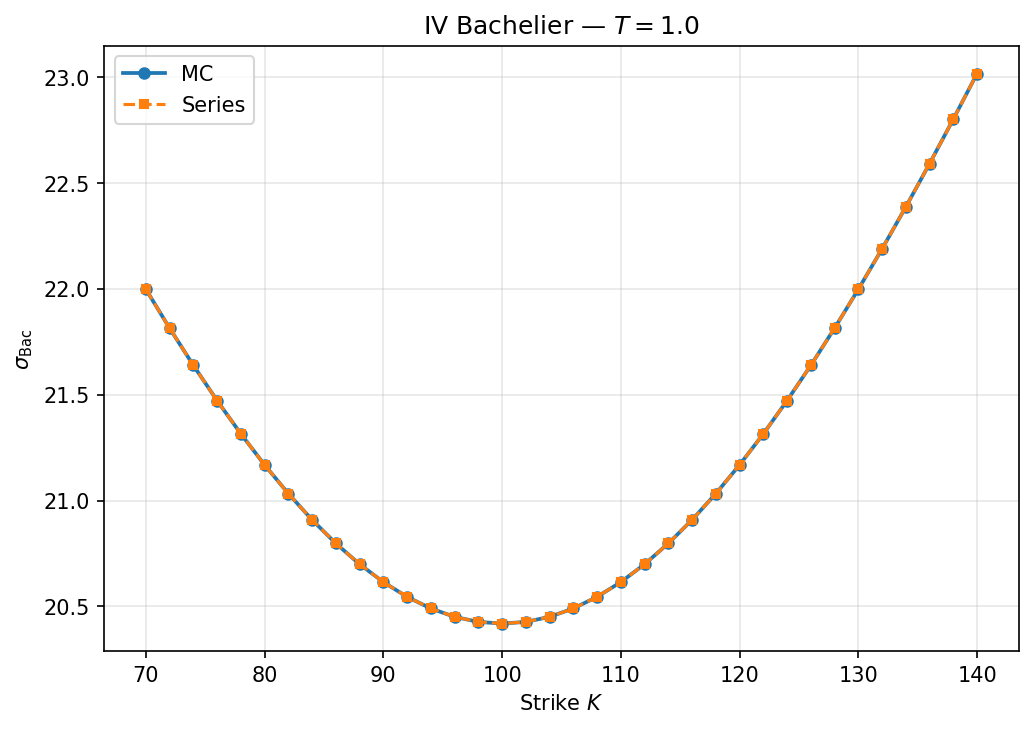}
    \end{minipage}

    \vspace{0.4cm}

    \begin{minipage}{0.48\textwidth}
        \centering
        \includegraphics[width=\textwidth]{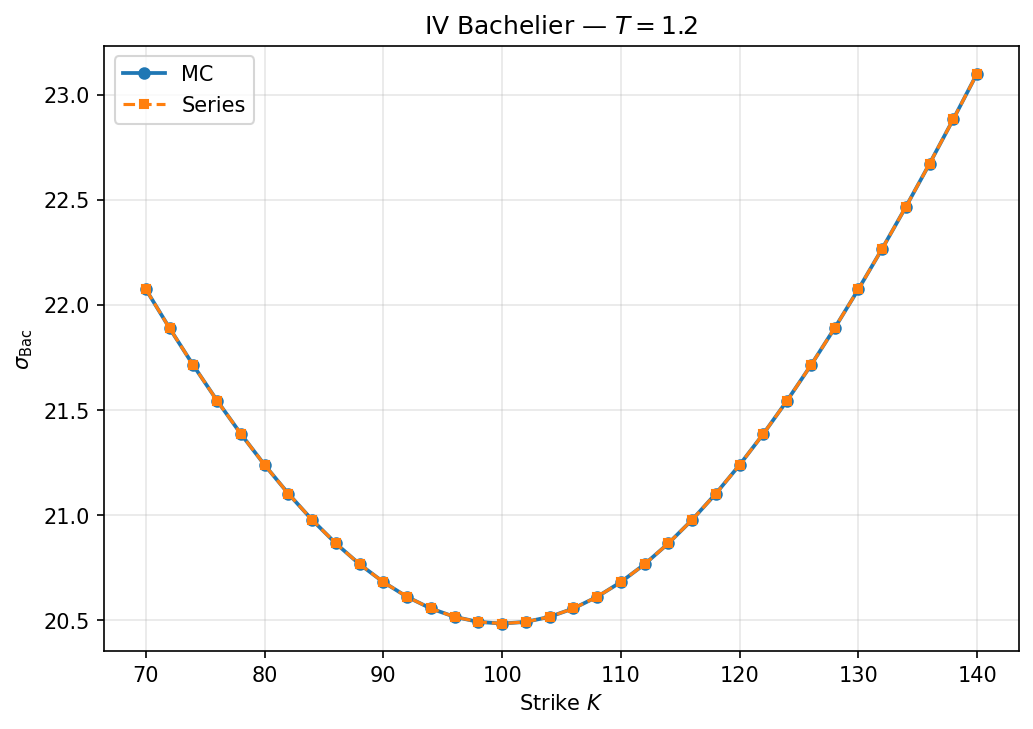}
    \end{minipage}

    \caption{Approximation of implied volatilities for the SABR model}
    \label{fig:SABRivs}
\end{figure}
As it can be observed, the fitting of the implied volatility smiles is, again, very precise. To give an idea of the accuracy of our method, in tables \ref{tab:iv_error_T08}, \ref{tab:iv_error_T10}, and \ref{tab:iv_error_T12} we display the relative error between the benchmark and the approximated implied volatilities in order to show that the implied volatilities obtained by the call prices computed as in Theorem \ref{principal} provide a great fit.

\begin{table}[H]
    \centering
    \begin{tabular}{SS | SS}
        \toprule
        {Strike} & {Rel. Error} & {Strike} & {Rel. Error} \\
        \midrule
        70  & 1.413e-9  & 106 & 3.550e-11 \\
        72  & 1.156e-9  & 108 & 6.412e-11 \\
        74  & 9.427e-10 & 110 & 1.022e-10 \\
        76  & 7.622e-10 & 112 & 1.508e-10 \\
        78  & 6.100e-10 & 114 & 2.113e-10 \\
        80  & 4.819e-10 & 116 & 2.853e-10 \\
        82  & 3.747e-10 & 118 & 3.747e-10 \\
        84  & 2.853e-10 & 120 & 4.819e-10 \\
        86  & 2.113e-10 & 122 & 6.100e-10 \\
        88  & 1.508e-10 & 124 & 7.622e-10 \\
        90  & 1.022e-10 & 126 & 9.427e-10 \\
        92  & 6.412e-11 & 128 & 1.156e-9  \\
        94  & 3.550e-11 & 130 & 1.413e-9  \\
        96  & 1.560e-11 & 132 & 1.947e-9  \\
        98  & 3.874e-12 & 134 & 1.350e-8  \\
        100 & 0         & 136 & 4.430e-7  \\
        102 & 3.874e-12 & 138 & 1.403e-5  \\
        104 & 1.560e-11 & 140 & 3.773e-4  \\
        \bottomrule
    \end{tabular}
    \caption{Relative error in implied volatility for $T = 0.8$.}
    \label{tab:iv_error_T08}
\end{table}

\begin{table}[H]
    \centering
    \begin{tabular}{SS | SS}
        \toprule
        {Strike} & {Rel. Error} & {Strike} & {Rel. Error} \\
        \midrule
        70  & 1.252e-8  & 106 & 3.592e-10 \\
        72  & 1.045e-8  & 108 & 6.460e-10 \\
        74  & 8.658e-9  & 110 & 1.024e-9  \\
        76  & 7.106e-9  & 112 & 1.501e-9  \\
        78  & 5.766e-9  & 114 & 2.086e-9  \\
        80  & 4.613e-9  & 116 & 2.791e-9  \\
        82  & 3.628e-9  & 118 & 3.628e-9  \\
        84  & 2.791e-9  & 120 & 4.613e-9  \\
        86  & 2.086e-9  & 122 & 5.766e-9  \\
        88  & 1.501e-9  & 124 & 7.106e-9  \\
        90  & 1.024e-9  & 126 & 8.658e-9  \\
        92  & 6.460e-10 & 128 & 1.045e-8  \\
        94  & 3.592e-10 & 130 & 1.252e-8  \\
        96  & 1.584e-10 & 132 & 1.489e-8  \\
        98  & 3.939e-11 & 134 & 1.752e-8  \\
        100 & 0         & 136 & 1.716e-8  \\
        102 & 3.939e-11 & 138 & 8.665e-8  \\
        104 & 1.584e-10 & 140 & 2.873e-6  \\
        \bottomrule
    \end{tabular}
    \caption{Relative error in implied volatility for $T = 1.0$.}
    \label{tab:iv_error_T10}
\end{table}

\begin{table}[H]
    \centering
    \begin{tabular}{SS | SS}
        \toprule
        {Strike} & {Rel. Error} & {Strike} & {Rel. Error} \\
        \midrule
        70  & 5.912e-9  & 106 & 1.845e-10 \\
        72  & 4.991e-9  & 108 & 3.309e-10 \\
        74  & 4.177e-9  & 110 & 5.227e-10 \\
        76  & 3.461e-9  & 112 & 7.628e-10 \\
        78  & 2.833e-9  & 114 & 1.055e-9  \\
        80  & 2.286e-9  & 116 & 1.402e-9  \\
        82  & 1.811e-9  & 118 & 1.811e-9  \\
        84  & 1.402e-9  & 120 & 2.286e-9  \\
        86  & 1.055e-9  & 122 & 2.833e-9  \\
        88  & 7.628e-10 & 124 & 3.461e-9  \\
        90  & 5.227e-10 & 126 & 4.177e-9  \\
        92  & 3.309e-10 & 128 & 4.991e-9  \\
        94  & 1.845e-10 & 130 & 5.912e-9  \\
        96  & 8.150e-11 & 132 & 6.955e-9  \\
        98  & 2.030e-11 & 134 & 8.241e-9  \\
        100 & 0         & 136 & 1.370e-8  \\
        102 & 2.030e-11 & 138 & 1.400e-7  \\
        104 & 8.150e-11 & 140 & 3.316e-6  \\
        \bottomrule
    \end{tabular}
    \caption{Relative error in implied volatility for $T = 1.2$.}
    \label{tab:iv_error_T12}
\end{table}
\end{example}
\begin{example}[Computation of Greeks]
As it has been mentioned in Remark \ref{Greeks}, differentiating with respect to $X_0$ the expression derived in Theorem \ref{principal} provides an analytical way to compute the $\Delta$ and the $\Gamma$ of the options in a fast and accurate way. To show it, consider first the Heston model with the same set of parameters as in Example  4.1. We will compute the $\Delta$ of several options under the Bachelier Heston model with our method and we will compare it to the $\Delta$ obtained by finite differences with a step-size $h = 10^{-3}$. The benchmark is the $\Delta$ computed by differentiating the conditional Monte Carlo expectation. As it is seen in Figure \ref{fig:heston_delta}, the three methods provide an excellent fit of the $\Delta$ of the option. The difference between our method and the other two is the computational cost. In Table \ref{tab:computation_times} we see that our method is the fastest for the computation of the $\Delta$.
\begin{figure}[H]
    \centering
    \includegraphics[width=0.5\linewidth]{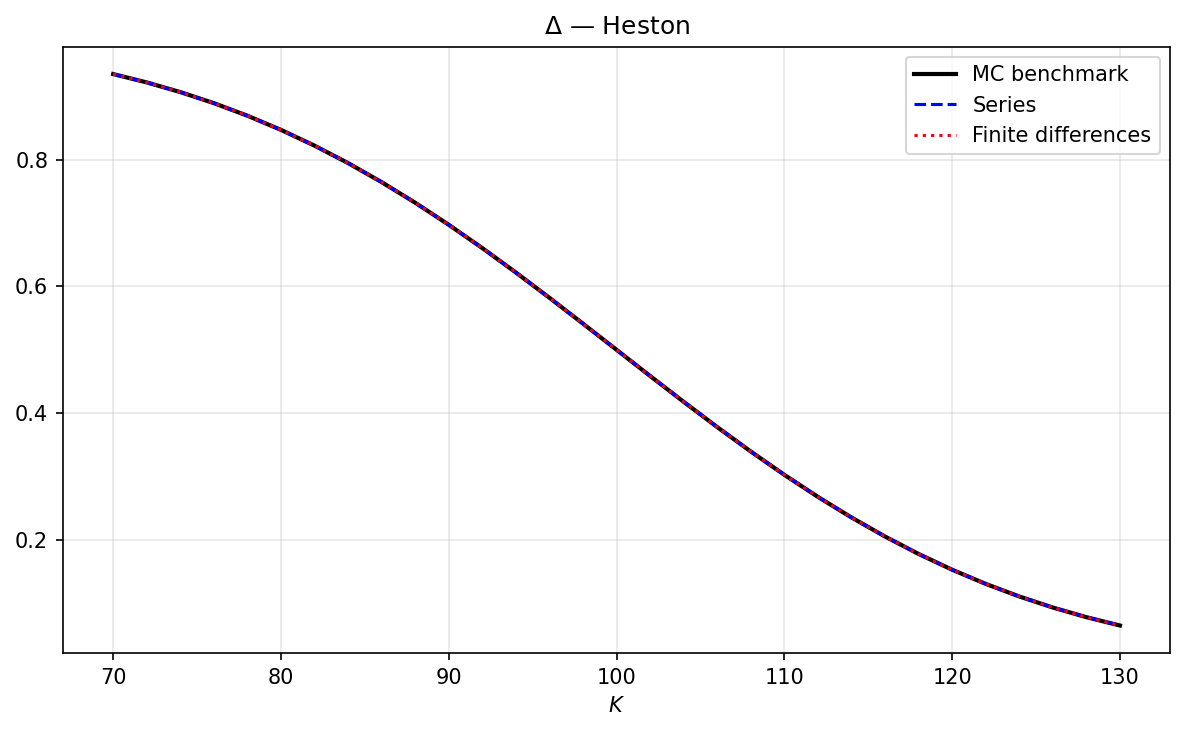 }
    \caption{Greek $\Delta$ computed by the 3 stated methods with $X_0 = 100$ and $T = 1.0$.}
    \label{fig:heston_delta}
\end{figure}

\begin{table}[H]
    \centering
    \begin{tabular}{lc}
        \toprule
        Method & Time (seconds) \\
        \midrule
        Benchmark MC       & 2.49  \\
        \textbf{Series}             & \textbf{0.42}  \\
        Finite Differences & 46.42 \\
        \bottomrule
    \end{tabular}
    \caption{Computation times for the different methods.}
    \label{tab:computation_times}
\end{table}
A similar phenomenon happens with the computation of Gamma, in this case we consider the Bachelier SABR model with $X_0 = 2$, $\sigma_0 = 0.7$ and $\nu = 0.3$ for the sake of diversity. In figure \ref{fig:Gamma_sabr_ir} we observe that again the fit provided by the 3 methods is excellent. Table \ref{tab:computation_times_sabr_ir} shows again that our method over-performs the other two in computational speed. 
\begin{figure}[H]
    \centering
    \includegraphics[width=0.5\linewidth]{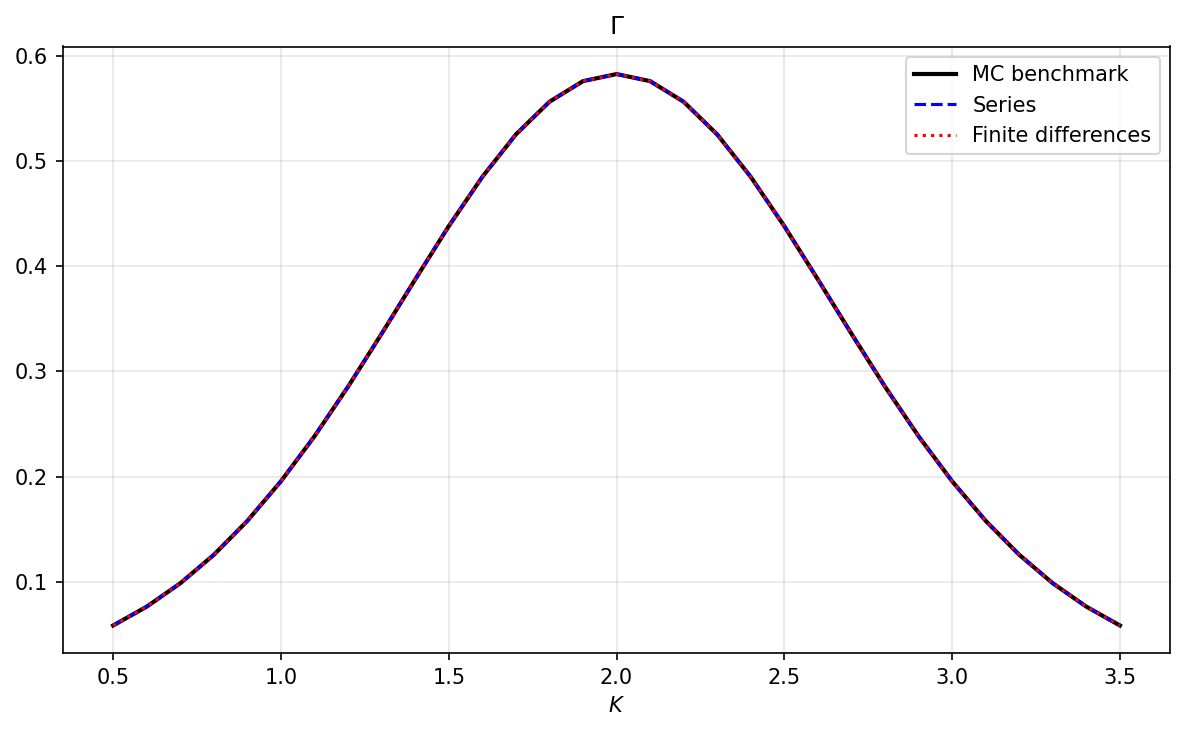}
    \caption{Greek $\Gamma$ computed by the 3 stated methods with $X_0 = 2$ and $T = 1.0$}
    \label{fig:Gamma_sabr_ir}
\end{figure}
\begin{table}[H]
    \centering
    \begin{tabular}{lc}
        \toprule
        \textbf{Method} & \textbf{Time (s)} \\
        \midrule
        Benchmark MC       & 3.92   \\
        \textbf{Series}             & \textbf{0.95}   \\
        Finite Differences & 125.25 \\
        \bottomrule
    \end{tabular}
    \caption{Computation times for the different methods}
    \label{tab:computation_times_sabr_ir}
\end{table}
\end{example}
\begin{example}[Monte Carlo Variance Reduction]
Another iinteresting quality of the expansion provided in Theorem \ref{principal} is that, for certain option, it works as a great control variate. To show the variance reduction, we will consider three different models used for computing options:
\begin{enumerate}
    \item [\textbf{(I)}] \textbf{Bachelier Heston} model with $X_0 = 100$, $\sigma_0 = 20$, $\kappa = 2$, $\theta = 400$, $\nu = 20$ and $\rho =-0.3$.
    \item [\textbf{(II)}] \textbf{Bachelier SABR} model with $X_0 = 100$, $\sigma_0 = 20$, $\nu = 0.5$ and $\rho = -0.5$.
    \item [\textbf{(III)}] \textbf{Bachelier SABR} model with $X_0 = 2$, $\sigma_0 = 0.7$, $\nu = 0.3$ and $\rho = -0.3$.
\end{enumerate}
As control variates, we will study the variance reduction of the following choices:
\begin{itemize}
    \item [\textbf{(CV1)}] A linear control variate $X_T - X_0$ where $X$ follows one of the models \textbf{(I)}--\textbf{(III)}.
    \item [\textbf{(CV2)}] A control variate based on the variance swap, that is,
    \[
    \frac{1}{T}\int_0^T \sigma_s^2 ds - E\left[\frac{1}{T}\int_0^T \sigma_s^2 ds \right].
    \]
    \item [\textbf{(CV3)}] A control variate based on the volatility swap, that is,
    \[
    \sqrt{\frac{1}{T}\int_0^T \sigma_s^2 ds} - E\left[\sqrt{\frac{1}{T}\int_0^T \sigma_s^2 ds }\right].
    \]
     \item [\textbf{(CV4)}] A control variate based on the expansion given in Theorem \ref{principal}, that is,
     \[
     (X_T^0 - K)_+ - V,
     \]
     where $X^0$ denotes one of the models \textbf{(I)}--\textbf{(III)} with $\rho = 0$ and $V$ is the price of the option under model $X^0$ computed via the expansion given in Theorem \ref{principal}.
\end{itemize}
For every control variate $Z$ selected between \textbf{(CV1)}--\textbf{(CV4)} we will find $\beta^*$ such that
\[
\beta^* = \arg \min_{\beta} \Var \left( (X_T - K)_+ - \beta Z\right).
\]
In order to highlight the variance reduction, we will plot the following two quantities:
\[
 \Var \left( (X_T - K)_+ - \beta^* Z\right), \quad \frac{\Var((X_T - K)_+)}{ \Var \left( (X_T - K)_+ - \beta^* Z \right)}.
\]
\begin{figure}[H]
    \centering
    \begin{minipage}{0.48\textwidth}
        \centering
        \includegraphics[width=\textwidth]{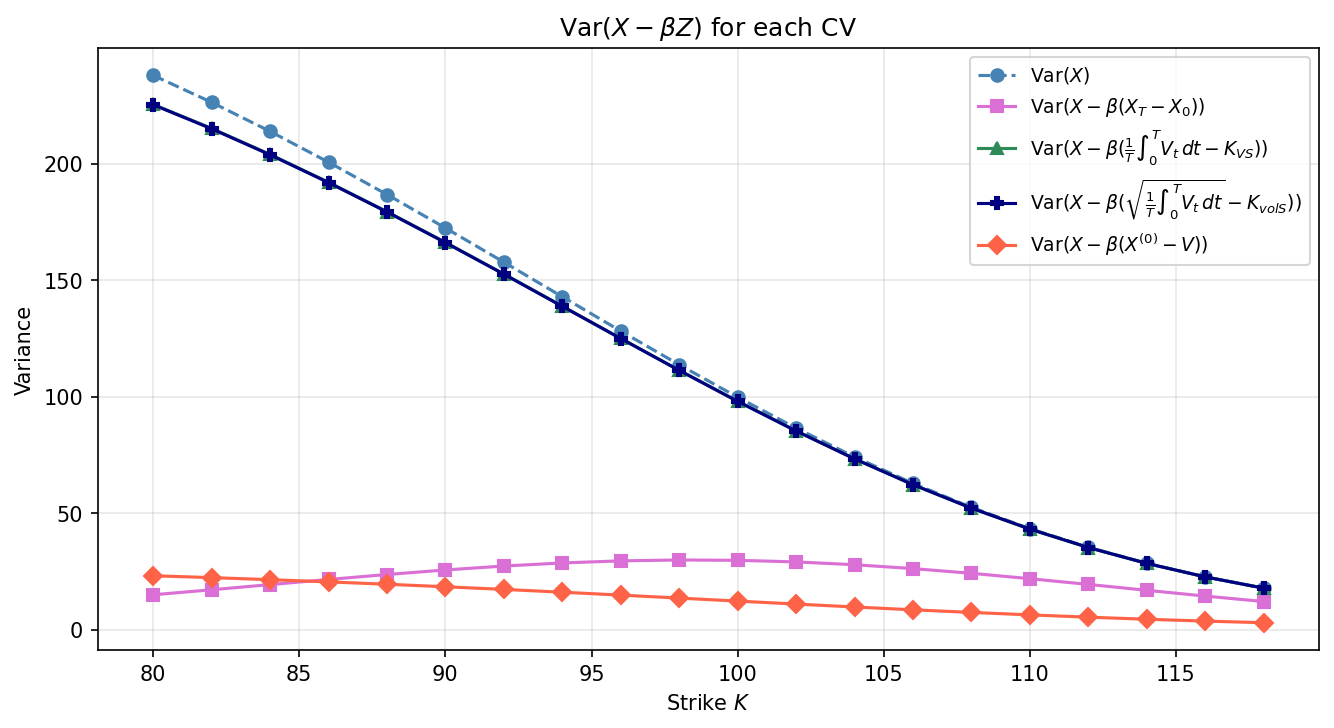}
    \end{minipage}
    \hfill
    \begin{minipage}{0.48\textwidth}
        \centering
        \includegraphics[width=\textwidth]{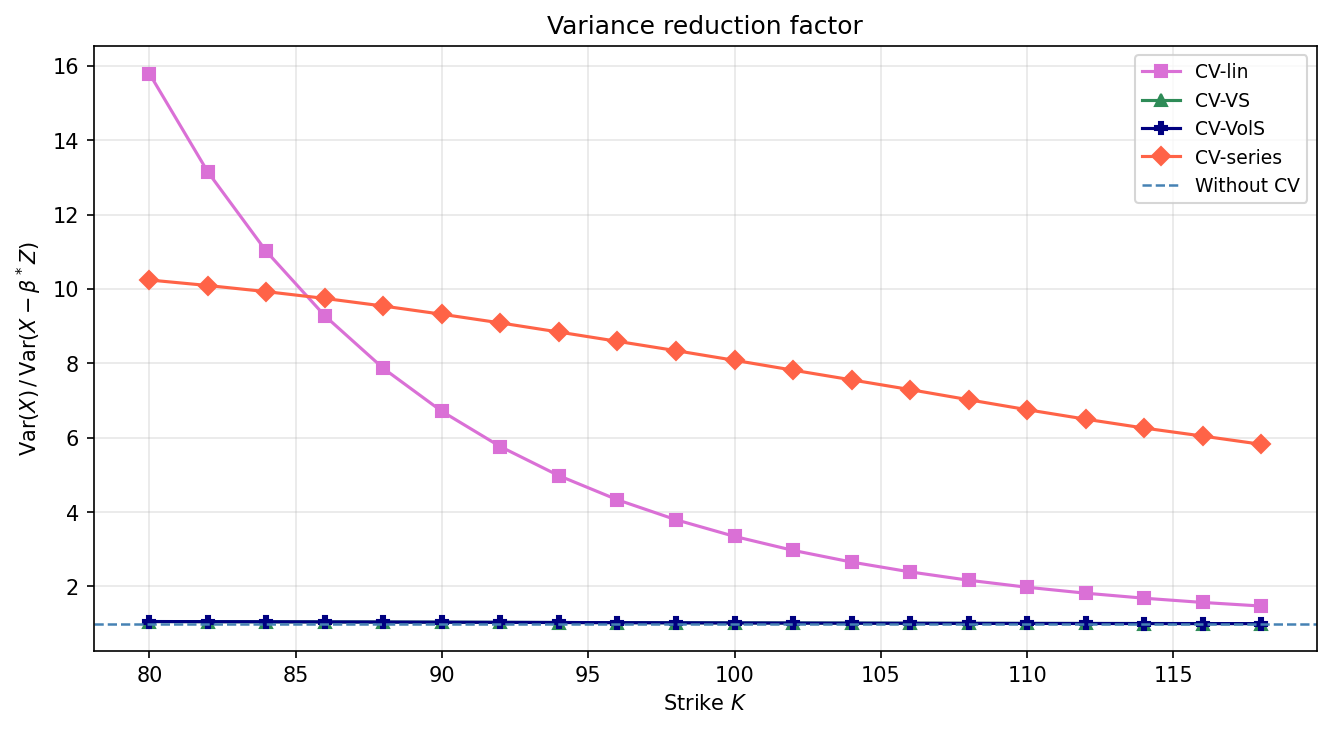}
    \end{minipage}
    \caption{Variance and variance reduction factor for each control variate in model \textbf{(I)}}
    \label{fig:cv_heston}
\end{figure}
\begin{figure}[H]
    \centering
    \begin{minipage}{0.48\textwidth}
        \centering
        \includegraphics[width=\textwidth]{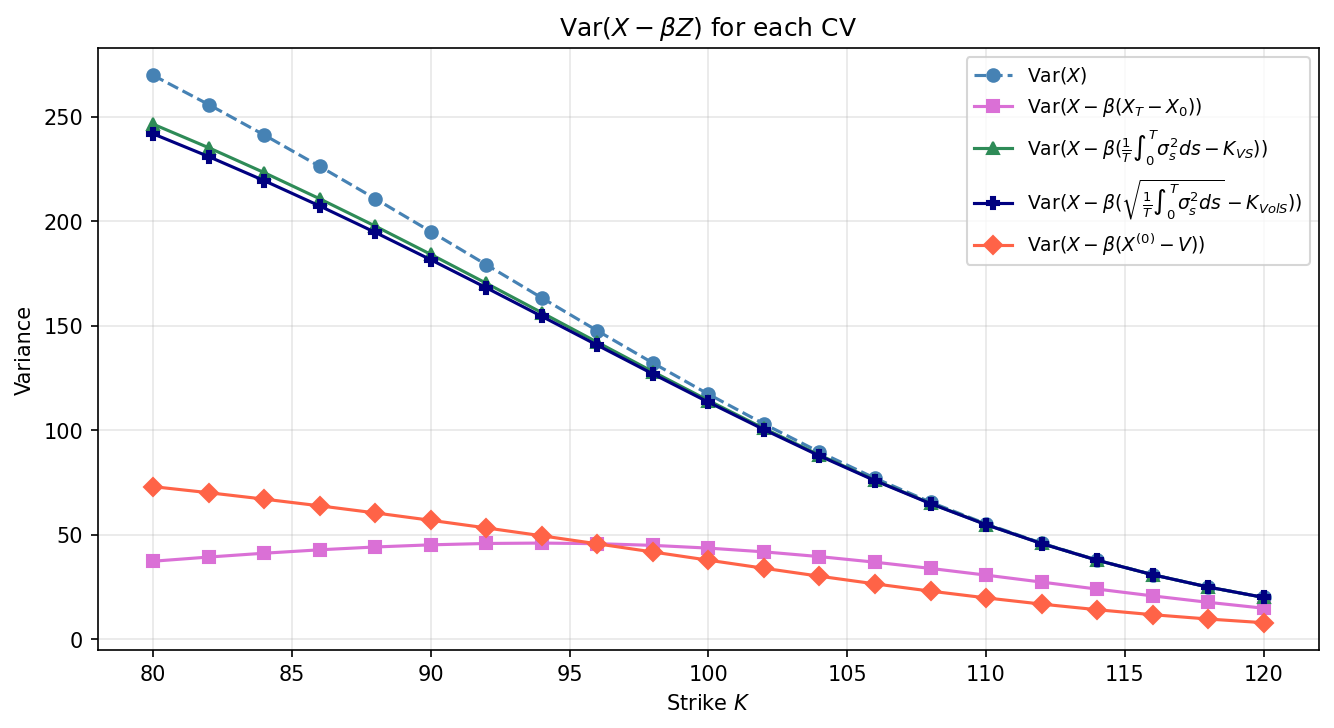}
    \end{minipage}
    \hfill
    \begin{minipage}{0.48\textwidth}
        \centering
        \includegraphics[width=\textwidth]{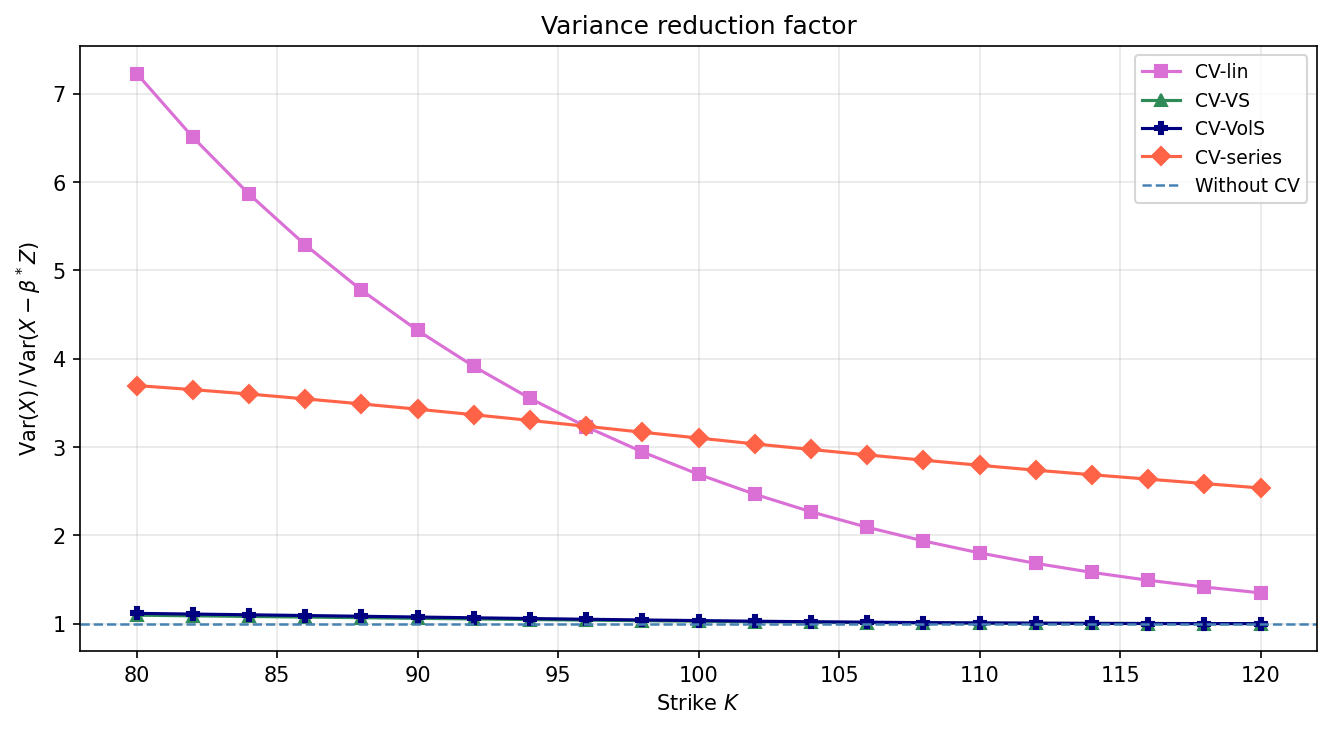}
    \end{minipage}
    \caption{Variance and variance reduction factor for each control variate in model \textbf{(II)}}
    \label{fig:cv_sabr}
\end{figure}

\begin{figure}[H]
    \centering
    \begin{minipage}{0.48\textwidth}
        \centering
        \includegraphics[width=\textwidth]{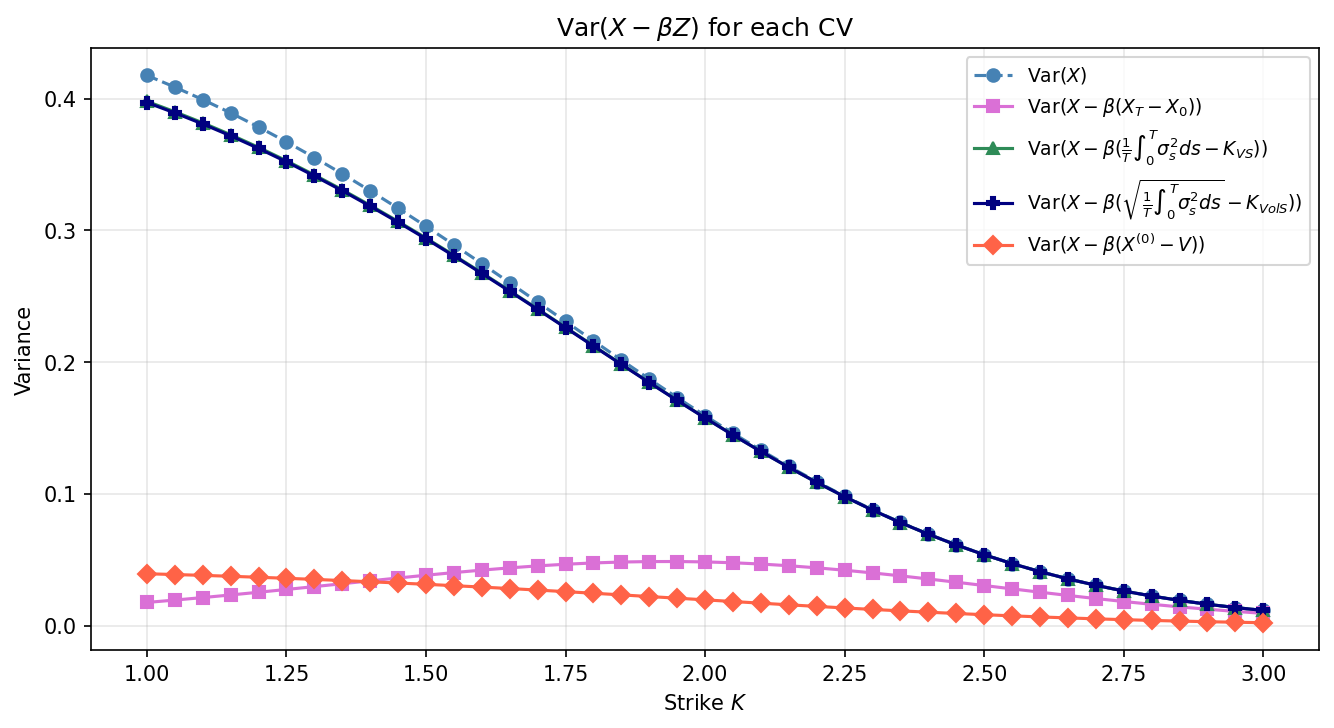}
    \end{minipage}
    \hfill
    \begin{minipage}{0.48\textwidth}
        \centering
        \includegraphics[width=\textwidth]{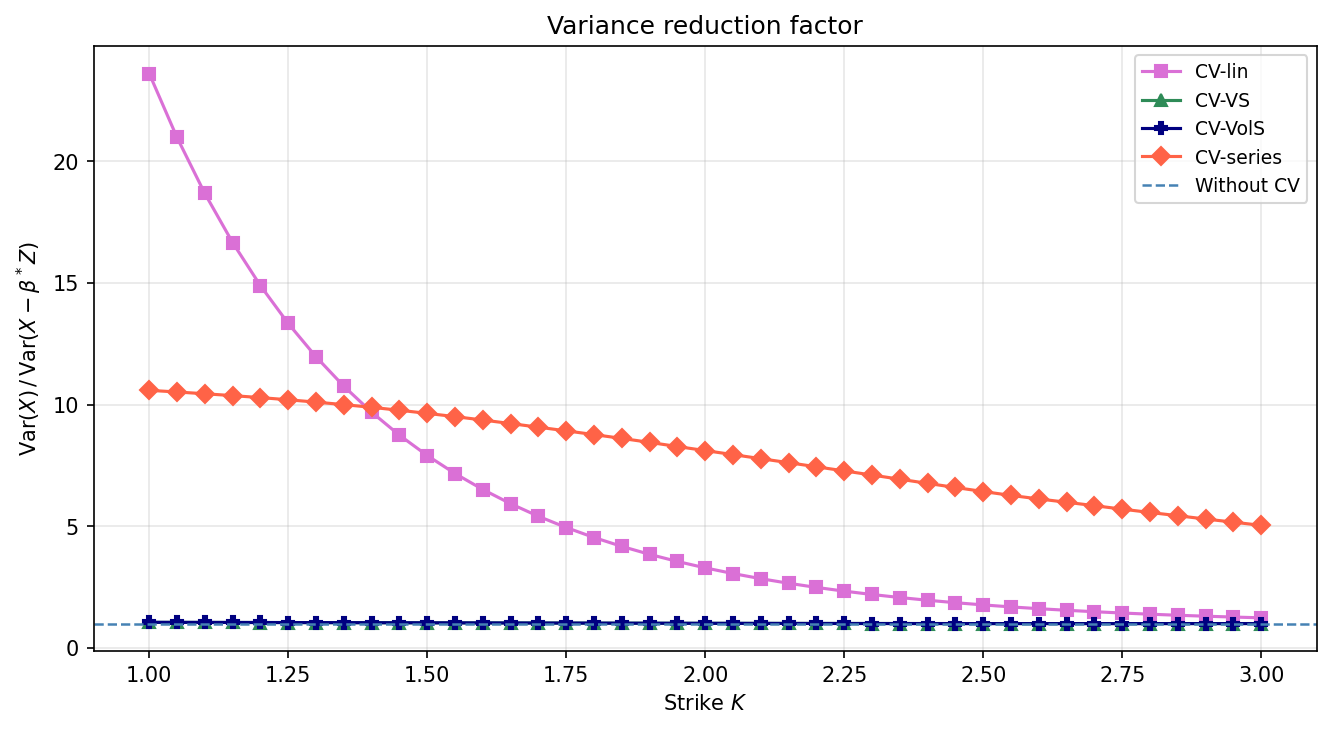}
    \end{minipage}
    \caption{Variance and variance reduction factor for each control variate in model \textbf{(III)}}
    \label{fig:cv_sabr2}
\end{figure}
In Figures \ref{fig:cv_heston}, \ref{fig:cv_sabr} and \ref{fig:cv_sabr2} we see that our control variate, \textbf{(CV4)} outperforms the other control variates in the OTM regime. Near ATM our control variate works better when $\rho = -0.3$. In fact, it is expected that the performance of our control variate decreases as $|\rho| \to 1$. In the deep ITM regime, since the payoff satisfies $(X_T - K)_+ \approx X_T - K$, it is natural that the linear control variate $X_T - X_0$ is the one that exhibits the major variance reduction.
\end{example}

\bibliographystyle{apalike}
\bibliography{references.bib}

@article{baviera2025smile,
  title={Smile asymptotic for Bachelier Implied Volatility},
  author={Baviera, Roberto and Massaria, Michele Domenico},
  journal={arXiv preprint arXiv:2506.08067},
  year={2025}
}

@article{unkwnown,
  title={On the Bachelier implied volatility at extreme strikes},
  author={Floc'h, Fabien Le},
  journal={arXiv preprint arXiv:2211.10232},
  year={2022}
}

@article{alos2025short,
  title={Short-time behavior of the At-The-Money implied volatility for the jump-diffusion stochastic volatility Bachelier model},
  author={Al{\`o}s, Elisa and Bur{\'e}s, {\`O}scar and Vives, Josep},
  journal={arXiv preprint arXiv:2503.22282},
  year={2025}
}

@article{alos2012decomposition,
  title={A decomposition formula for option prices in the Heston model and applications to option pricing approximation},
  author={Al{\`o}s, Elisa},
  journal={Finance and Stochastics},
  volume={16},
  number={3},
  pages={403--422},
  year={2012},
  publisher={Springer}
}

@article{ASENS_1900_3_17__21_0,
     author = {Bachelier, L.},
     title = {Th\'eorie de la sp\'eculation},
     journal = {Annales scientifiques de l'\'Ecole Normale Sup\'erieure},
     pages = {21--86},
     year = {1900},
     publisher = {Elsevier},
     volume = {3e s{\'e}rie, 17},
     doi = {10.24033/asens.476},
     language = {fr},
     url = {https://www.numdam.org/articles/10.24033/asens.476/}
}

@article{BergomiGuyon2012,
  author  = {Lorenzo Bergomi and Julien Guyon},
  title   = {Stochastic Volatility's Orderly Smiles},
  journal = {Risk Magazine},
  year    = {2012},
  month   = {May},
  pages   = {60--66}
}

@article{choi2022black,
  author    = {Jaehyuk Choi and Minsuk Kwak and Chyng Wen Tee and Yumeng Wang},
  title     = {{A Black--Scholes user's guide to the Bachelier model}},
  journal   = {Journal of Futures Markets},
  volume    = {42},
  number    = {5},
  pages     = {959--980},
  year      = {2022},
  publisher = {Wiley Online Library}
}

@book{fouque2000derivatives,
  title={Derivatives in financial markets with stochastic volatility},
  author={Fouque, Jean-Pierre and Papanicolaou, George and Sircar, K Ronnie},
  year={2000},
  publisher={Cambridge University Press}
}

@article{fouque2,
  author = {Fouque, J. P. and Papanicolaou, G. and Sircar, R. and Solna, K.},
  title = {Singular Perturbations in Option Pricing},
  journal = {SIAM Journal on Applied Mathematics},
  volume = {63},
  number = {5},
  pages = {1648-1665},
  year = {2003},
  doi = {10.1137/S0036139902401550}
}

@article{hagan2002managing,
  title={Managing smile risk},
  author={Hagan, Patrick S and Kumar, Deep and Lesniewski, Andrew S and Woodward, Diana E},
  journal={The Best of Wilmott},
  volume={1},
  number={1},
  pages={249--296},
  year={2002}
}

@article{fukasawa2011asymptotic,
  title={{Asymptotic analysis for stochastic volatility: martingale expansion}},
  author={Fukasawa, Masaaki},
  journal={Finance and Stochastics},
  volume={15},
  pages={635--654},
  year={2011},
  publisher={Springer}
}

@article{Lewis02092022,
author = {Alan L. Lewis and Dan Pirjol},
title = {Proof of non-convergence of the short-maturity expansion for the SABR model},
journal = {Quantitative Finance},
volume = {22},
number = {9},
pages = {1747--1757},
year = {2022},
publisher = {Routledge},
doi = {10.1080/14697688.2022.2071759}}

@book {lewis2000option,
    AUTHOR = {Lewis, Alan L.},
     TITLE = {Option valuation under stochastic volatility. {II}},
      NOTE = {With Mathematica code},
 PUBLISHER = {Finance Press, Newport Beach, CA},
      YEAR = {2016},
     PAGES = {viii+737},
      ISBN = {978-0-9676372-1-1; 0-9676372-1-X},
   MRCLASS = {91-02 (60H30 60J75 62P05 65Cxx 65Mxx 91B70 91G20)},
  MRNUMBER = {3526206},
}

@article {antonelli2006pricing,
    AUTHOR = {Antonelli, Fabio and Scarlatti, Sergio},
     TITLE = {Pricing options under stochastic volatility: a power series
              approach},
   JOURNAL = {Finance Stoch.},
  FJOURNAL = {Finance and Stochastics},
    VOLUME = {13},
      YEAR = {2009},
    NUMBER = {2},
     PAGES = {269--303},
      ISSN = {0949-2984,1432-1122},
   MRCLASS = {91Gxx},
  MRNUMBER = {2482054},
       DOI = {10.1007/s00780-008-0086-4},
       URL = {https://doi.org/10.1007/s00780-008-0086-4},
}

@article {alos2020exponentiation,
    AUTHOR = {Al\`os, Elisa and Gatheral, Jim and Radoi\v{c}i\'{c},
              Rado\v{s}},
     TITLE = {Exponentiation of conditional expectations under stochastic
              volatility},
   JOURNAL = {Quant. Finance},
  FJOURNAL = {Quantitative Finance},
    VOLUME = {20},
      YEAR = {2020},
    NUMBER = {1},
     PAGES = {13--27},
      ISSN = {1469-7688,1469-7696},
   MRCLASS = {91G20 (60H30 60H99 91B70)},
  MRNUMBER = {4040259},
MRREVIEWER = {Piotr\ Nowak},
       DOI = {10.1080/14697688.2019.1642506},
       URL = {https://doi.org/10.1080/14697688.2019.1642506},
}

\end{document}